\documentclass[aps,twocolumn,pre,showpacs,showkeys,reprint]{revtex4}

\usepackage{amsmath,amsfonts,amssymb,amsthm}
\usepackage[usenames]{color}

\usepackage{bm}

\usepackage{xspace}

\usepackage{paralist}

\usepackage{graphicx}

\usepackage{subfig}

\usepackage[colorlinks,urlcolor=blue,citecolor=blue,linkcolor=blue]{hyperref}

\usepackage{algorithm}
\usepackage{setspace}

\newtheorem{theorem}{Theorem}
\newtheorem{claim}[theorem]{Claim}

\newtheorem{example_hidden}[theorem]{Example}

\newcommand{\cM}{{\cal M}}

\newcommand{\eps}{\varepsilon}

\newcommand{\Sec}[1]{\hyperref[sec:#1]{\S\ref*{sec:#1}}} 
\newcommand{\Eqn}[1]{\hyperref[eq:#1]{(\ref*{eq:#1})}} 
\newcommand{\Fig}[1]{\hyperref[fig:#1]{Fig.\,\ref*{fig:#1}}} 
\newcommand{\Tab}[1]{\hyperref[tab:#1]{Tab.\,\ref*{tab:#1}}} 
\newcommand{\Thm}[1]{\hyperref[thm:#1]{Thm.\,\ref*{thm:#1}}} 
\newcommand{\Lem}[1]{\hyperref[lem:#1]{Lem.\,\ref*{lem:#1}}} 
\newcommand{\Prop}[1]{\hyperref[prop:#1]{Prop.~\ref*{prop:#1}}} 
\newcommand{\Cor}[1]{\hyperref[cor:#1]{Cor.~\ref*{cor:#1}}} 
\newcommand{\Def}[1]{\hyperref[def:#1]{Defn.~\ref*{def:#1}}} 
\newcommand{\Alg}[1]{\hyperref[alg:#1]{Alg.~\ref*{alg:#1}}} 
\newcommand{\Ex}[1]{\hyperref[ex:#1]{Ex.~\ref*{ex:#1}}} 
\newcommand{\Clm}[1]{\hyperref[clm:#1]{Claim~\ref*{clm:#1}}} 

\newcommand{\M}[1]{{\bm{\mathbf{\MakeUppercase{#1}}}}} 

\newcommand{\bv}  {{\mathbf v}}
\newcommand{\bp}  {{\mathbf p}}
\newcommand{\be}  {{\mathbf e}}
\newcommand{\sta} {\boldsymbol{\pi}}

\begin{document}

\title{A stopping criterion for Markov chains when generating independent random graphs}
\author{J. Ray, A. Pinar and C. Seshadhri}
\affiliation{Sandia National Laboratories, Livermore, CA, USA}

\date{\today}


\begin{abstract}

Markov chains are  convenient means of generating realizations of
networks with a given (joint or otherwise) degree distribution, since
they simply require a procedure for rewiring edges. The major
challenge is to find the right number of steps to run such a chain, so
that we generate truly independent samples. Theoretical bounds for
mixing times of these Markov chains are too large to be practically
useful. Practitioners have no useful guide for choosing the length, and
tend to pick numbers fairly arbitrarily.  We give a principled
mathematical argument showing that it suffices for the length to be
proportional to the number of desired number of edges. We also
prescribe a method for choosing this proportionality constant.  We run
a series of experiments showing that the distributions of common graph
properties converge in this time, providing empirical evidence for our
claims.


  \keywords{graph generation, Markov chain Monte Carlo, independent samples}
  \pacs{89.75.Hc,05.10.Gg}
\end{abstract}

\maketitle

\section{Introduction}
\label{sec:intro}
Graphs are a common topological representation across a variety of
scientific fields. They are used when relations between a large number
of entities have to be specified in a succinct manner. Chemical
reactions, molecules, social networks (both physical and online) and
the electric grid are some common examples. In many cases, we may have
partial information about the graph, requiring generation of many
graphical realizations consistent with the (partial) characterization.
Such situations arise in case of large online social networks (of
which only a small part can be sampled tractably) or where the data
simply cannot be collected e.g., the web of human sexual relations
(which only allows the estimation of the degree distribution). In
other cases, privacy concerns can prevent the distribution of a graph
for experimentation and or study, e.g., networks of email
communications, critical infrastructure nets, etc. This gives rise to
the problem of constructing ``sanitized'' proxies, that preserve only
some properties of the original graph. Thus, being able to generate
independent graphs conditioned on an incomplete set of graphical
measurements is essential for many applications.

Many graph generation methods aim to preserve the salient features of
graphs~\cite{12sk3a, 12pf4a,55hv1a,62hs1a,96wp2a, kronecker,
Aiello00,Chung2, 12gk3a}.  For statistical analysis, we need
algorithms that can generate uniformly random instances from the space
of graphs with a specified feature.  There has been significant work
on random graph generation with a given degree distribution (DD),
which specifies the number of vertices with a given degree.
In~\cite{ktv}, the problem of generating a graph with a given degree
distribution was reduced to a prefect matching problem, which can be
used to generate random instances by employing results in sampling
perfect matchings~\cite{JeSiVi04,Br86}.  Alternatively, sequential
sampling methods were investigated in~\cite{BlDi11,BaKiSa10}.  These
methods can be compute intensive, and in practice,
%
%
Markov chains (MC) are widely used due to their simplicity and
flexibility.  The MC is started using a graph that honors the
specified graphical characteristics, and uses a procedure that can
``rewire'' a graph, to create ensembles of graphs with the same degree
distribution.  Taylor~\cite{81tr1a} showed that edge-swaps could
modify a graph while preserving its DD.
Kannan et al~\cite{ktv} analyzed the mixing time of such a MC, whereas
Gkantsidis~\cite{GkantsidisMMZ03} devised a MC scheme that avoids self
loops. In~\cite{StPi12}, Stanton and Pinar used an MC to generate an
ensemble of graphs using a rewiring scheme that preserved the joint
degree distribution (JDD), which specifies the number of edges between
vertices of specified degrees.  Stanton and Pinar also empirically
analyzed the mixing of the MC using the binary ``time-series'' traced
out by the appearance/disappearance of edges between two labeled
nodes, as the MC executed its random walk in the space of graphs. They
showed that the autocorrelation of the time-series decayed to zero, a
necessary condition for the MC to converge to its stationary
distribution~\cite{Sokal}.



Graph rewiring schemes that preserve DD or JDD are simple (and fast)
and MC chains driving them are easy to construct. However, successive
graphs generated by the MC are only slightly different and the MC has
to proceed for a large number of steps $N$ before the initial graph is
``forgotten.'' The mixing time estimates in~\cite{ktv} take the form
of upper bounds and of $O(|E|^6)$, where $|E|$ is the number of edges
in a graph. Even for small graphs with 1000 edges, the bound is
intractable. In practice, $N$ is chosen arbitrarily. Failure to mix
completely leads to the generation of correlated samples and any
results derived from them are erroneous.  While the method
in~\cite{StPi12} demonstrated the use of autocorrelation to establish
mixing, it is an \emph{a posteriori} test which provides no \emph{a
priori} guidance regarding the length of the MC chain. Further, the
method requires a small, user-specified threshold, below which the
time-series autocorrelation is deemed zero.

In this paper, we construct analytical models that estimate $N$. These
models track the evolution of the binary ``time-series'' formed by the
edges of labeled graphs generated by the MC.  We test the model under
conditions where the DD or the JDD is held constant. In
Sec.~\ref{sec:models}, the model is used to derive an expression for
the mixing time based on when the binary time-series begins to
resemble independent draws from a distribution. The models predict
that the mixing time is proportional to $|E|$, the number of edges in
the graph. The model holds true for a ``representative'' edge and
exploits the constancy of DD (or JDD) in arriving at its
prediction. The model is tested with real graphs in
Sec.~\ref{sec:tests}. In Sec.~\ref{sec:verif}, we develop a
data-driven method, that assumes neither a constant DD (or JDD) nor a
representative edge, to investigate the independence of the edge
time-series. We use this test to verify the assumptions made when
developing mixing time expressions in Sec.~\ref{sec:models}. We
determine the fraction of the edges for which the model prediction of
$N$ is an underestimate, and the consequences of the lack of
stationarity on the ensemble of graphs sampled in this manner.


%
\section{Theoretical analysis}
\label{sec:models}

The goal of this paper is to provide a mathematically principled
argument for running a MC ${\rm O(|E|)}$ steps to generate independent
graphs with $|E|$ edges. The constant hidden in the big-Oh depends on
the desired accuracy. The graphs so generated may have a prescribed DD
or a JDD; we find that a MC on graphs with a prescribed DD mixes
slightly easier than those where the JDD is preserved. Our empirical results
show that ${\rm 5|E|}$ -- ${\rm 30|E|}$ steps are sufficient for
mixing these MCs. We provide a mathematical justification
for this observation.

Consider a MC on the space of graphs and two labeled vertices $u$ and
$v$. Under certain circumstances, the existence of an edge $(u,v)$ can
also be described by a Markov chain. Thus we model the behavior of the
MC on the space of graphs in terms of a set of coupled 2-state Markov
chains, each representing edge behavior. We develop the transition
matrix for the (edge) MC, which is exact for DD-preserving
experiments, but heuristically derived for MCs with prescribed
JDDs. We prove an ${\rm O(|E|)}$ mixing time for these 2-state Markov
chains.  This ensures that that after ${\rm O(|E|)}$ MC steps, each
edge behaves as if it comes out of independent draws. While this is
not sufficient for complete mixing, it is a \emph{necessary}
condition.

These mathematical theorems will be empirically validated in the next section,
but we give a quick preview of the method.
Starting from a real
graph, we run a MC (driving a DD- or JDD-preserving rewiring scheme)
for $\ell$ steps, to create a new, nominally independent graph. This
is repeated $M$ times to create an ensemble of graphs. We plot the
distribution of some common graphical quantity (e.g., the number of
triangles in each graph) constructed from the ensemble. This is
repeated for different values of $\ell$ till the plots converge. We
find that $\ell \sim {\rm 10 |E|}$. When $\ell \sim {\rm |E|}$, the MC
is far from convergence.

Below, we review the rewiring schemes used in the paper. This is
followed by derivations of the models of mixing under DD and JDD
preservation.

\subsection{A review of rewiring schemes}
\label{sec:rewire}

Consider an undirected graph $G=(V,E)$, where $|V| = n$, the number of
vertices in the graph, and $|E| = m$, the number of edges.  The degree
distribution of the graph is given by the vector $\mathbf{f}$, where
$\mathbf{f}(d)$ is the number of vertices of degree $d$.  The \emph{joint
degree distribution} is an $n \times n$ matrix $\mathbf{J}$, where
$\mathbf{J}(i,j)$ is the number of edges incident between vertices of degree
$i$ and degree $j$. We use $d_v$ to denote the degree of vertex $v$.

 
\begin{figure}
  \includegraphics[scale=0.5]{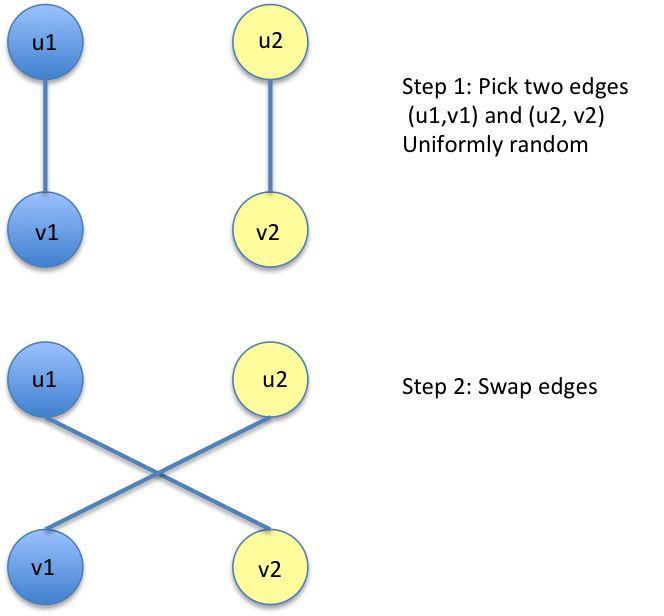}
  \caption{The DD-preserving swapping operation described in
  Sec.~\ref{sec:rewire}.}
  \label{fig:ddswap}
\end{figure}

In this paper, we use the greedy, JDD-preserving rewiring procedure by
Stanton and Pinar in~\cite{StPi12}. This is analogous to the degree
distribution preserving chain~\cite{ktv,GkantsidisMMZ03}, illustrated in Fig~\ref{fig:ddswap}.  
The JDD-preserving rewiring
procedure is shown in Fig~\ref{fig:jddswap}. The procedure does not guard
against parallel edges or self-loops, but the MC itself can be
slightly modified to ensure that these artifacts do not
occur~\cite{StPi12,GkantsidisMMZ03}. We also maintain lists of nodes
and edges indexed by their degree, so that for a given degree $d$, we
can locate a uniform random edge that is incident on a degree-$d$
node. The steps in the rewiring scheme are:

\begin{asparaitem}
  \item Pick an \emph{endpoint} uniformly at random. This is done by
  picking a random edge and choosing one of its endpoints with
  equal probability. Let this vertex be $u_1$. We then choose
  a uniform random neighbor $v$ of $u_1$, to get edge
  $(u_1,v)$. 

  \item Pick uniformly at random another vertex $u_2$ such that $d_{u_1} =
  d_{u_2}$, and choose a uniform random neighbor $w$ incident to
  $u_2$. This gives the second edge.

  \item Swap the edges $(u_1,v)$ and $(u_2,w)$ to create $(u_1,w)$ and
  $(u_2,v)$. 
\end{asparaitem}
Details of the rewiring scheme and discussions of the ergodicity,
mixing of a MC driving such a rewiring scheme can be found
in~\cite{StPi12}. A pictorial description is in Fig.~\ref{fig:jddswap}.

\begin{figure}
  \includegraphics[scale=0.5]{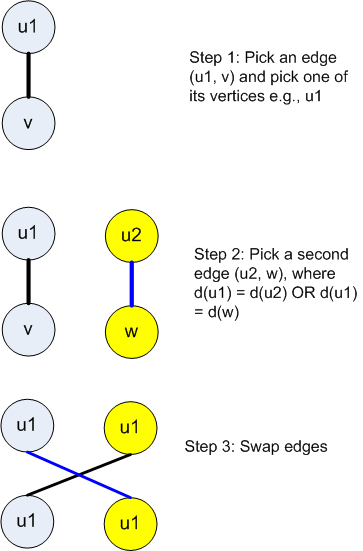}
  \caption{The JDD-preserving swapping operation described in
  Sec.~\ref{sec:rewire}.}
  \label{fig:jddswap}
\end{figure}

\subsection{Modeling mixing when degree distribution is preserved}
\label{sec:DD}

Consider a MC $\cM$ that generates graphs with a prescribed degree
distribution. Consider two labeled vertices $u$ and $v$, with degrees
$d_u$ and $d_v$. We will show that the appearance/disappearance of $(u,v)$ in  $\cM$ can also be
described by a Markov chain.

\begin{claim}
\label{clm:T}
  The transition matrix of the Markov chain for the edge $(u,v)$ is given by
  \begin{eqnarray}
    \M{T}(d_u, d_v)  & = &  \left( \begin{array}{cc}           
                                  1 - \alpha(d_u,d_v) &   \alpha(d_u,d_v) \\
                                  \beta (d_u,d_v)     &  1 - \beta(d_u,d_v) 
                               \end{array} \right)  \nonumber 
  \label{eqn:T}
  \end{eqnarray}
  where  $\alpha(d_u,d_v)=\frac{d_ud_v}{2m^2}$,   $\beta (d_u,d_v)=1-(1-1/m)^2$, and   the first index in $T$ is state $0$ (no edge).
\end{claim}

\begin{proof}
  Assume that the edge $(u,v)$ does not initially exist. We analyze
  the probability of it appearing after one step. Let $e$ and
  $e^{\prime}$ be the two edges chosen by the rewiring algorithm. $(u,
  v)$ will appear if $u$ is the first endpoints of $e$ and $v$
  is the last endpoint of $e'$ (or vice versa). This leads to four
  possible cases. Consider the case when both the vertices are the
  first endpoints of $e$ and $e^{\prime}$. The probability of this
  event is $d_u/2m \times d_v/2m = d_u d_v / 4m^2$. The other case
  is symmetric and so $\alpha(d_u, d_v) = d_u d_v / 2m^2$.

  Suppose that $(u, v)$ does exist but is removed during the Markov
  transition. This occurs if $(u, v)$ is either $e$ or
  $e^{\prime}$. The probability of this event is $\beta_(d_u, d_v) = 1 - (1-1/m)^2$. Note
  that the values of $\alpha$ and $\beta$ so derived are completely
  independent of the structure of the rest of the graph, and so the
  transition probabilities in Eq.~\ref{eqn:T} satisfy the Markov property.
\end{proof}

\subsection{Modeling mixing when joint degree distribution is preserved}
\label{sec:JDD}

We consider a pair of labeled vertices $(u, v)$ and model the
appearance/disappearance of the edge during a Markov transition.
For convenience, we will assume that $\mathbf{f}(d_u)$ and $\mathbf{f}(d_v)$
are both strictly greater than $1$. The proofs can be extended to these
cases as well. (Observe that when $\mathbf{f}(d_u) = \mathbf{f}(d_v)$, the
edge $(u,v)$ is always present because the JDD is preserved.)

\begin{claim}
\label{clm:beta}
 Assume edge $(u, v)$ is present and is removed using the transition. The probability of  this event is
 \begin{equation}
   \beta(d_u, d_v) = \frac{1}{m} + \frac{\mathbf{f}(d_u) -  1}{2m\mathbf{f}(d_u)} + \frac{\mathbf{f}(d_v) -  1}{2m\mathbf{f}(d_v)}
 \end{equation}
\end{claim}
\begin{proof}
  We follow the JDD-preserving swapping procedure described in
  Sec.~\ref{sec:rewire}. Let $e$ and $e^{\prime}$ be the two edges
  chosen for swapping. If   $(u, v)$ is $e$, it will definitely be
  swapped. The probability of the chosen edge being $(u, v)$ is $1/m$.

  The other alternative is that  $(u, v)$ is $e^{\prime}$. 
  This can happen only if the first endpoint chosen has degree $d_u$ (or $d_v$)
  and then $(u,v)$ is chosen. Consider the first case.
%
  The total number
  of endpoints on degree $d_u$  vertices (excluding $u$) is
  $(\mathbf{f}(d_u)-1)d_u$. The probability that $e$ is chosen
  is the probability that $u$ is chosen as the second endpoint, and
  then $(u,v)$ is chosen. This is equal to $1/(\mathbf{f}(d_u) \times d_u)$.
  The total probability is 
  \[  \frac{(\mathbf{f}(d_u) - 1)d_u}{2m} \times \frac{1}{\mathbf{f}(d_u)d_u} = \frac{(\mathbf{f}(d_u) -  1)d_u}{2m\mathbf{f}(d_u)d_u}  = \frac{\mathbf{f}(d_u) -  1}{2m\mathbf{f}(d_u)}\]

  Symmetrically, the randomly chosen endpoint could have had degree
  $d_v$. So the total probability of choosing $e^{\prime} = (u, v)$ is
$$ \frac{\mathbf{f}(d_u) -  1}{2m\mathbf{f}(d_u)} + \frac{\mathbf{f}(d_v) -  1}{2m\mathbf{f}(d_v)} $$    
  and the probability of $(u, v)$ being removed is
$$
     \beta(d_u, d_v) = \frac{1}{m} + \frac{\mathbf{f}(d_u) -  1}{2m\mathbf{f}(d_u)} + \frac{\mathbf{f}(d_v) -  1}{2m\mathbf{f}(d_v)}.
$$
\end{proof}

This expression has certain ramifications. Since the JDD is preserved
(and by implication, $\mathbf{f}(d)$), the probability of edge removal does not
vary as the Markov chain proceeds. This satisfies the Markov property
and $\beta$ depends only on $u$ and $v$ i.e., it is independent of the
rest of the graph.

Consider the probability $\alpha$ of $(u, v)$ appearing. This is not
Markovian, as it depends on the structure of the graph. However, this
dependence is weak, and we can obtain a Markovian
approximation using a simple heuristic.

Consider the number of edges $\M{J}(d, d_v)$ between vertices of
degree $d_v$ and an arbitrary $d$. The \emph{expected} number of such
edges incident on a degree $d_v$ vertex is $\M{J}(d,d_v)/\mathbf{f}(d_v)$. The
number of edges between a vertex $v$ and degree $d$ vertices depends
on the graphical structure; however, we will approximate it with the
expected value $\M{J}(d,d_v)/\mathbf{f}(d_v)$.

\begin{claim}
\label{clm:alpha}
  If $(u, v)$ is not present at a stage in the Markov chain, the
  probability $\alpha$ that it appears in the next step is given by
  \begin{equation}
    \alpha(d_u, d_v) = \frac{2\M{J}(d_u,d_v)}{m\mathbf{f}(d_u)\mathbf{f}(d_v)}
  \label{eqn:alpha}
  \end{equation}
\end{claim}

\begin{proof} As before, let the two edges chosen by the algorithm
  be $e$ and $e'$, in that order. The edge $(u,v)$ is swapped in 
  when either of the four following situations happens.
  Vertex  $u$ is the first (resp. second) endpoint of $e$ and $v$ is the 
  second (resp. first) endpoint of $e'$. Or, 
  vertex  $u$ is the first (resp. second) endpoint of $e'$ and $v$ is the 
  second (resp. first) endpoint of $e$.
  
  $u$ is the first endpoint of $e$ and $v$ is the second endpoint of $e'$ (or vice versa).
  
  Consider the following sequence of events. First, $u$ is chosen as an endpoint
  (and we get edge $e$). Then, we choose a vertex $u'$ (who is a neighbor
  of $v$) whose degree is the same as $u$.
  Finally, we select edge $e' = (u',v)$. Swapping will lead to edge $(u,v)$.
  The probability of the first event is $d_u/2m$. The probability of the 
  second event depends the number of neighbors of degree $d_u$ incident
  to $v$. Alternately, this is the number of edges connecting $v$
  to degree $d_u$ vertices. Clearly, this depends on the graph structure, but our approximation of this 
  quantity is $\M{J}(d_u,d_v)/\mathbf{f}(d_v)$.
  
  So, the second probability is $\M{J}(d_u,d_v)/(\mathbf{f}(d_v)\mathbf{f}(d_u))$ (because
  there are $\mathbf{f}(d_u)$ neighbors of degree $d_u$, of which $\M{J}(d_u,d_v)/\mathbf{f}(d_v)$
  are neighbors of $v$). The last probability is $1/d_u$. The total probability is
  \[ \frac{d_u}{2m} \times \frac{\M{J}(d_u,d_v)}{\mathbf{f}(d_u)\mathbf{f}(d_v)} 
  \times \frac{1}{d_u}= \frac{\M{J}(d_u,d_v)}{2m\mathbf{f}(d_u)\mathbf{f}(d_v)} \]
  
  Now consider the second case. First, $v'$ (a neighbor of $u$ whose degree is $d_v$) is chosen
  as an endpoint. Then, $u$ is chosen as the neighbor of $v'$, so $e = (u,v')$.
  Then, $v$ is chosen as the first endpoint of the second edge $e'$.
  Swapping will again lead to edge $(u,v)$. The first probability 
  is approximated by $\mathbf{J}(d_u,d_v)/(\mathbf{f}(d_u)\mathbf{f}(d_v))$.
  The second probability is $1/d_v$, and the last is $d_v/2m$.
  The probability comes out to be the same as earlier.
  
  By symmetry, we get the probabilities of the other cases by switching the roles
  of $u$ and $v$, yielding the same expression. 
%
%
%
%
\end{proof}

\subsection{Estimating the mixing time}

Having developed expressions for $\alpha$ and $\beta$ in
Eq.~\ref{eqn:T}, under the assumption of prescribed DD
(Sec.~\ref{sec:DD}) and JDD (Sec.~\ref{sec:JDD}), we use the Markov
transition matrix to estimate the number of steps $N$ to achieve a
stationary distribution, conditional on a tolerance $\epsilon$.

\begin{claim} 
\label{clm:markov} Choose $N$ as follows.
Let the final distribution after
running $\cM_{u,v}$ for $N$ steps be $\bp$. Then $\|\bp - \sta_{u,v}\|
< \eps$.
\[ N=\begin{cases}
(m/2)\ln(1/\eps) & \mathrm{DD\;preserving\;\; MCs} \\
m\ln(1/\eps) &   \mathrm{JDD\; preserving\;\; MCs }
\end{cases}\] 
\end{claim}

The distribution $\sta_{u,v}$ is the probability of existence/absence
of edge $(u, v)$ once $\cM$ has converged. This implies that after $N
= m\ln(1/\eps)$ we are very close to the converged distribution of
\emph{all} edges, irrespective of the degrees of their terminating
vertices.

\begin{proof}
$\M{T}$ has two eigenvalues, $1$ and $1 - (\alpha(d_u, d_v)
  +\beta(d_u, d_v))$. Let the corresponding unit eigenvectors be
  $\be_1$ and $\be_2$. For notational convenience we will refer to
  $\alpha(d_u, d_v)$ and $\beta(d_u, d_v)$ as $\alpha$ and
  $\beta$. Since $\alpha + \beta > 0$, the eigenvectors form a
  basis. We represent the initial state of the MC as $\bv = c_1 \be_1
  + c_2 \be_2$. After $N$ applications of the transition matrix, the
  distribution of the Markov chain's state is
\begin{eqnarray}
 \bp & = & \M{T}^N \bv = c_1 \M{T}^N \be_1 + c_2 \M{T}^N \be_2 \nonumber \\
     & = & c_1 \be_1 + c_2 \left(1 - (\alpha +\beta) \right)^N \be_2. 
\end{eqnarray}

Since $(1 - (\alpha + \beta)) < 1$, the second term decays with $N$,
and the stationary distribution asymptotically approaches $c_1
\be_1$. The 2-norm of the decaying term (the error incurred when we
stop at finite $N$) can be bounded for
$$N = \ln{(1/\epsilon)}/(\alpha+\beta)$$
in the following manner
\begin{eqnarray}
  \|\bp - \sta_{u,v}\|  & = &  \|(1-\gamma)^N c_2 \be_2\|_2 
           \leq  (1-\gamma)^{\ln(1/\eps)/\gamma}  c_2 \|\be_2\|_2  \nonumber \\
                        & \leq &  \exp(-\ln(1/\eps)) = \eps \nonumber
\end{eqnarray}
where we have used $\gamma = (\alpha + \beta)$ and 
\begin{eqnarray}
  N & = & \frac{m}{2} \ln{\frac{1}{\epsilon}} \geq \frac{1}{\gamma}  \ln{\frac{1}{\epsilon}},
  \mbox{\hspace{3mm} for DD-preserving MC and} \nonumber \\
  N & = & m \ln{\frac{1}{\epsilon}} \geq \frac{1}{\gamma} \ln{\frac{1}{\epsilon}},
  \mbox{\hspace{3mm} for JDD-preserving MC.} 
\label{eqn:N}
\end{eqnarray}
The key observation is $\gamma \geq 2/m$ for DD-preserving MC (see
values of $\alpha$ and $\beta$ in Sec.~\ref{sec:DD}) and $\gamma \geq
1/m$ for a JDD-preserving MC (see Sec.~\ref{sec:JDD}). 
\end{proof}

\section{Tests with real graphs}
\label{sec:tests}

In this section we perform  empirical studies to find a suitable value of $\epsilon$, i.e., one
which can be used to generate samples of graphs which are
statistically indistinguishable from those generated by a smaller
$\epsilon$. We discriminate graphs using certain graphical metrics. 

We run the MC with the DD-preserving rewiring scheme for $N$ steps
before saving the resultant graph. We start the MC using 4 real
graphs - the neural network of \emph{C. Elegans}~\cite{celegans1}
(referred to as ``C. Elegans''), the power grid of the Western states
of US~\cite{celegans1} (called ``Power''), co-authorship graph of
network science researchers~\cite{adjnoun} (referred to as
``Netscience'') and a 75,000 vertex graph of the social network at
Epinions.com~\cite{03ra3a} (``soc-Epinions1''). Their details are in
Table~\ref{tab:graphs}. The first three were obtained
from~\cite{newman_graphs} while the fourth was downloaded
from~\cite{snap}. All the graphs were converted to undirected graphs
by symmetrizing the edges.

We use $N = \{0.5, 2.5, 5, 7.5\}|E|$ corresponding to $\epsilon =
\{0.37, 6.7 \times 10^{-3}, 4.5 \times 10^{-5}, 3.06 \times
10^{-7}\}$. 1000 graphs are generated in this manner, starting with
the real graphs in Table~\ref{tab:graphs}. Thereafter, we calculate
the global clustering coefficient (three times the ratio of the number of
triangles to wedges in a graph)\cite{BaWe00},
the diameter, and the maximum eigenvalue of the Laplacian matrix  of the graph~\cite{Me94} 
 for each of members of the ensemble of 1000 graphs. In Fig.~\ref{fig:dd},
we plot these distributions for three of the four graphs in
Table~\ref{tab:graphs}. We clearly see that for $\epsilon = 0.37$
i.e., $N = 0.5|E|$, the distributions for all graphical metrics are
quite different from the distributions obtained with $N > |E|$. The
red and black curves, corresponding to $N = 5|E|$ and $N = 7.5|E|$ are
mostly indistinguishable indicating that the MC has converged to its
stationary distributions. We will henceforth proceed with $\epsilon =
4.5 \times 10^{-5}$, achieved with $N = 5|E|$ when the DD is preserved
during the MC sampling.
\begin{table}[b]
  \caption{Characteristics of the graphs used in this paper. $(|V|,
    |E|)$ are the numbers of vertices and edges in the graph and G-R
    diagnostic is the Gelman-Rubin diagnostic.}
  \label{tab:graphs}
  \begin{ruledtabular}
    \begin{tabular}{lcr} 
      \textrm{Graph name} & \textrm{ $(|V|, |E|)$} & \textrm{G-R diagnostic} \\ \colrule
      C. Elegans & (297, 4296)     &  1.05 \\ 
      Netscience & (1461, 5484)    &  1.02 \\ 
      Power      & (4941, 13188)   &  1.006 \\ 
      soc-Epinions1 & (75879, 405740)& 1.06  \\
    \end{tabular}
  \end{ruledtabular}
\end{table}

\begin{figure*}
  \centerline{\includegraphics[width=\textwidth, trim=1cm 7cm 2.cm  7cm,clip=true]
  {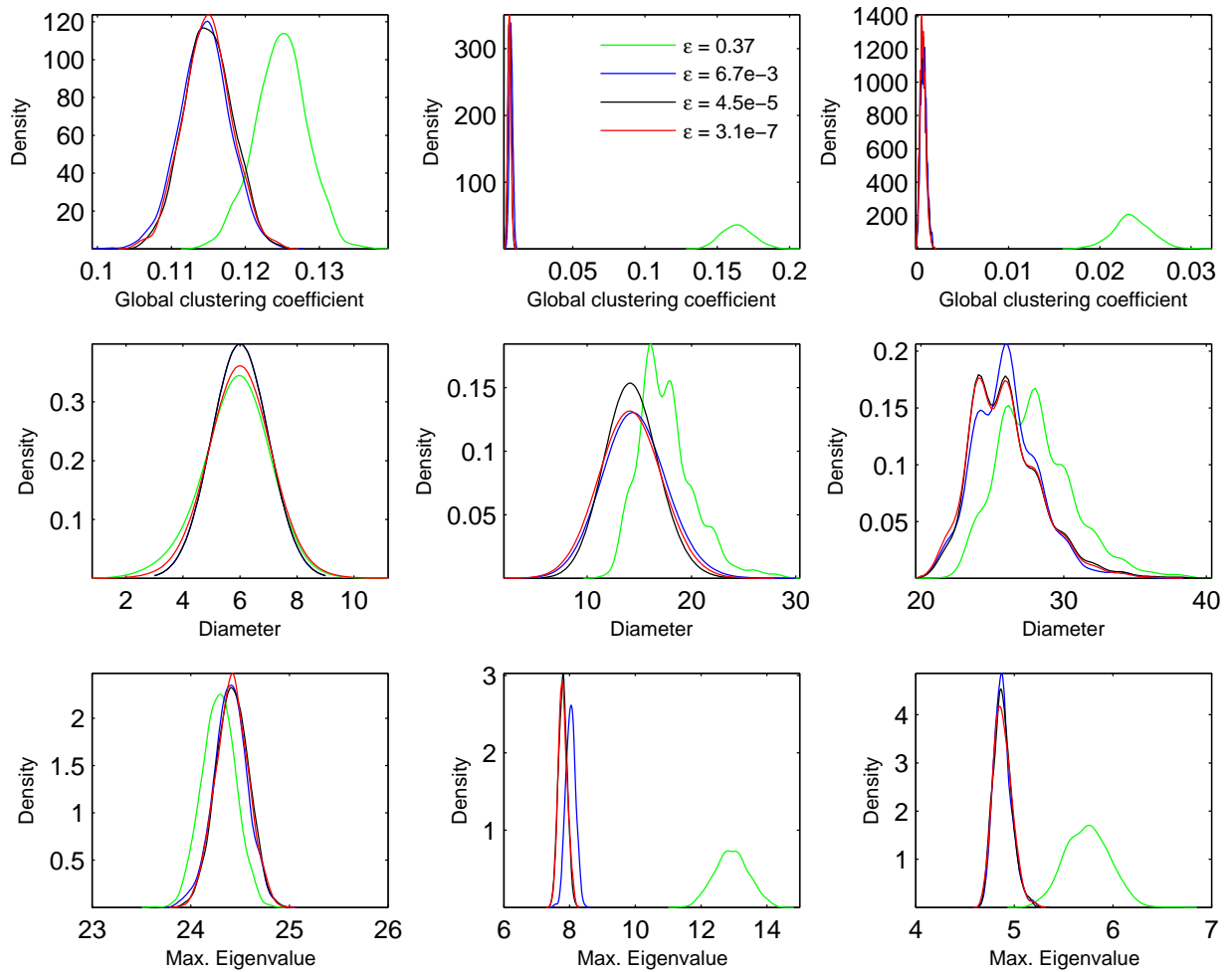}}
  \caption{Plots of the distributions of the global clustering
    coefficient, the diameter, and the maximum eigenvalue of the
    graph Laplacian for ``C. Elegans'' (left), ``Netscience'' (middle)
    and ``Power'' (right), evaluated after $0.5|E|, 2.5|E|, 5|E|$ and
    $7.5|E|$ iterations of the Markov chain (green, blue, black and
    red lines respectively). The corresponding values of $\epsilon$
    are in the legend. We see that the distributions converge at
    $\epsilon \sim 1e-5$. In these runs, the DD was held constant
    across all graphical samples.}
  \label{fig:dd}
\end{figure*}

We now turn our attention to the case when the JDD is preserved during
MC sampling. We repeat the same sampling procedure as above, but with
the JDD-preserving rewiring scheme being driven by the Markov
chain. We choose the same sequence of $\epsilon$, corresponding to $N
= \{1, 5, 10, 15\}|E|$, per Eq.~\ref{eqn:N}. The distributions of
global clustering coefficients, graph diameter and maximum eigenvalue
of the graph Laplacian are in Fig.~\ref{fig:jdd}. Again, it is
 clear that the distribution (for any of the graphical
metrics) is far from convergence for $\epsilon = 0.37$ but
indistinguishable for the other values of $\epsilon$. In particular,
the red and black plots lie on top of each other in most of the
subfigures, indicating that $N = 10|E|$ ($\epsilon = 4.5 \times
10^{-5}$) is generally sufficient to drive the MC to its stationary
distribution, as characterized by the aforementioned graphical
parameters.
\begin{figure*}
  \centerline{\includegraphics[width=\textwidth, trim=1cm 7cm 2.cm  7cm,clip=true]
  {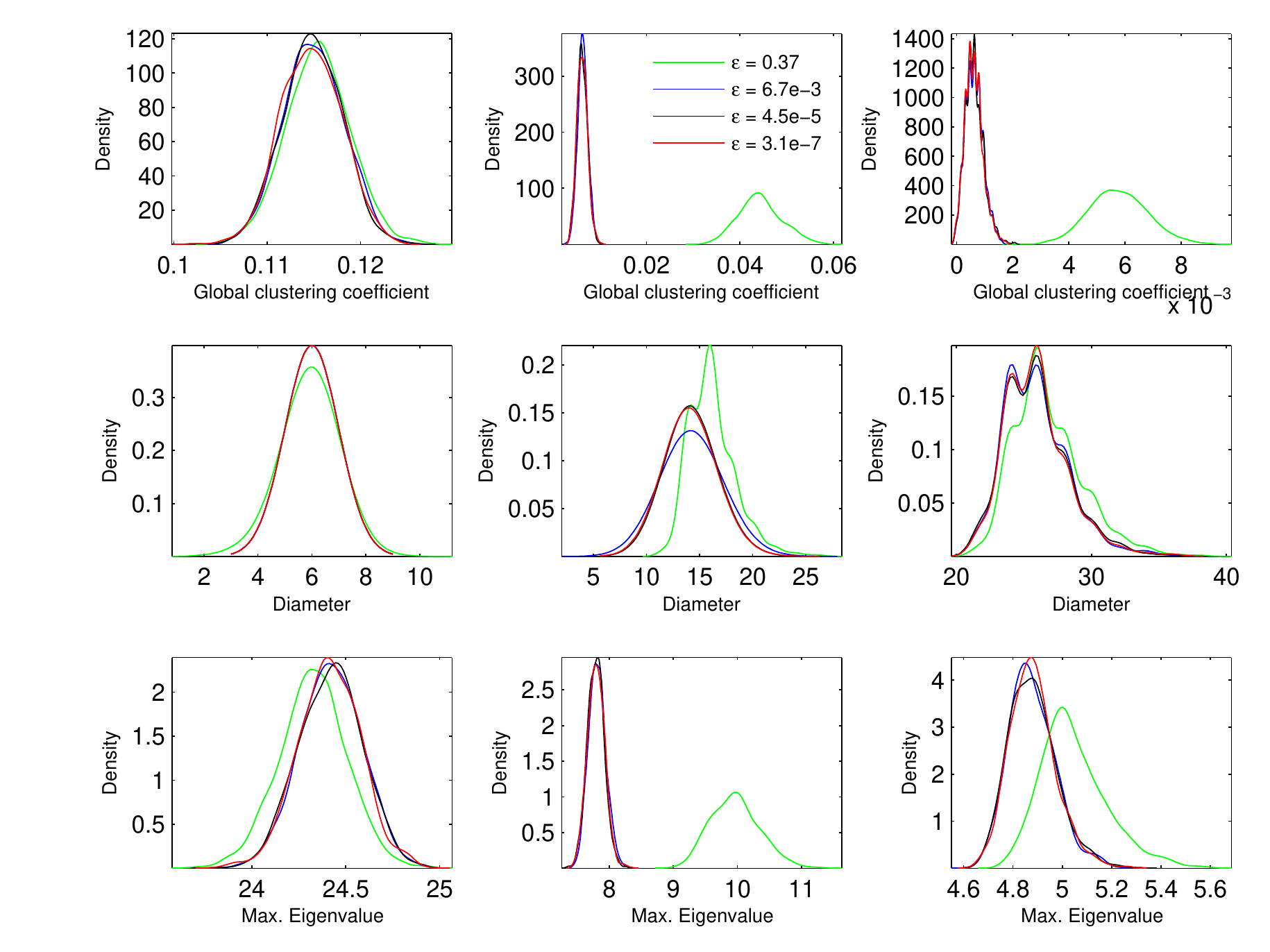}}
  \caption{Plots of the distributions of the global clustering
    coefficient, the diameter, and the maximum eigenvalue of the
    graph Laplacian for ``C. Elegans'' (left), ``Netscience'' (middle)
    and ``Power'' (right), evaluated after $|E|, 5|E|, 10|E|$ and
    $15|E|$ iterations of the Markov chain (green, blue, black and
    red lines respectively). The corresponding values of $\epsilon$
    are in the legend. We see that the distributions converge at
    $\epsilon \sim 1e-5$. In these runs, the JDD was held constant
    across all graphical samples.}
  \label{fig:jdd}
\end{figure*}

In summary, a tolerance of $\epsilon = 4.5 \times 10^{-5}$ provides
converged/stationary distributions of global clustering coefficient,
graph diameter and maximum eigenvalue of the graph Laplacian. This
corresponds to running the MC on the space of graphs for $N = 5|E|$
steps when DD is preserved and for $N = 10|E|$ steps when JDD is
preserved. These results provide an empirical proof that the models in
Sec.~\ref{sec:DD} and Sec.~\ref{sec:JDD} provide a good \emph{a
priori} estimate of the steps to run a MC to obtain stationary
results.

%
\section{Verification of model assumptions}
\label{sec:verif}

In this section, we check our assumption that a MC on the space of
graphs can be approximated by coupled 2-state Markov chains. We also
devise a test to ensure that the MC does not converge to a local mode
in the distribution of graphs. This can occur in case of multimodal
distributions.

\subsection{Verifying edge-by-edge convergence}
\label{sec:ebe}

The equation for $N$ (Eq.~\ref{eqn:N}) is based on a heuristic and has
to be verified. Further, the equation is strictly valid for an edge
and it is unlikely that after $N$ steps of the MC, \emph{all} the edges
have decorrelated and converged to their stationary distribution. The
residual number of correlated edges have to be identified and their
impact on graphical properties quantified.

Consider the behavior of an arbitrary edge (between labeled vertices)
during a long walk executed by a MC on the space of graphs. During
this walk, the edge appears and disappears, thus describing a binary
time-series $\{Z_t\}$. If the MC on the space of graphs reaches a
stationary distribution in ${\rm O(|E|)}$ steps, then thinning the
binary time-series by a factor of ${\rm O(|E|)}$ should result in a
time-series that resembles independent draws of zeros and ones. We
develop a nonparametric test of independence for such time-series.
This method exploits Sokal's observation that if an MC reaches a
stationary distribution, the samples it draws are uncorrelated i.e.,
the autocorrelation is small~\cite{Sokal}. It is general and does not
depend on the preservation of DD or JDD in the graphs. While this
test has been in existence for some time~\cite{92rl2a,96ra1a}, it is
not well known in the network generation literature.

Assume that the time-series $\{Z_t\}$ is very long i.e., it takes $K
\gg N$ steps. The time-series will be correlated as observed
in~\cite{StPi12}. Consider too, a $k$-thinned chain $\{Z^k_t\}$,
obtained by retaining every $k^{th}$ member of $\{Z_t\}$.  The thinned
chain will have a far smaller autocorrelation and upon further
thinning will begin to resemble independent draws from a
distribution. If Eq.~\ref{eqn:N} is correct, $k = N$ should yield a
time-series that resembles independent draws \emph{more} than a
first-order Markov process. This forms the basis of our test of
independence. We fit models of two processes, independence and
first-order Markov, to the data and compute the
log-likelihood. The ``goodness'' of model fit is determined by computing
the Bayesian Information Criterion (BIC). This paper is the first
application of this technique to graphs, though it has been used
in other contexts.

Let $x_{ij}$ be the number of $(i, j), i, j \in (0, 1)$ transitions in
the $k-$thinned chain $\{Z^k_t\}$. The $x_{ij}$ are used to populate
$X$, a $2 \times 2$ contingency table. The table entries are
normalized by the length of the $k$-thinned chain ($K/k -1$) to
provide us with the empirical probabilities $p_{ij}$ of an $(i,j)$
transition in $\{Z^k_t\}$. Let $\widehat{p_{ij}}$ and
$\widehat{x_{ij}} = (K/k -1) \widehat{p_{ij}}$ be the predictions
of expected values of the table entries provided by a model. Then the
goodness-of-fit of the model can be provided by a likelihood ratio
statistic (called the $G^2$-statistic; Chapter 4.2 in~\cite{07bfh3a})
and a Bayesian Information Criterion (BIC) score:
\begin{eqnarray}
  G^2 & = & -2 \sum_{i=0}^{i=1} \sum_{i=0}^{i=1} x_{ij} \log\left( \frac{\widehat{x_{ij}}}{x_{ij}}\right), 
  \nonumber \\
  BIC & = & G^2 + q \log\left( \frac{K}{k} - 1\right),
\label{eqn:ll}
\end{eqnarray}
where $q$ is the number of parameters in the model used to fit the
table data. Log-linear models are generally used to model tabular data
(Chapter 2.2.3 in~\cite{07bfh3a}). The log-linear models' predictions
for table entries generated by independent sampling and a first-order
Markov process are
\begin{eqnarray}
\log(p_{ij}^{(I)}) & = & u^{(I)} + u^{(I)}_{1,(i)} + u^{(I)}_{2, (j)} 
 \mbox{\hspace{2mm} and \hspace{2mm}}  \nonumber \\
\log(p_{ij}^{(M)}) & = & u^{(M)} + u^{(M)}_{1,(i)} + u^{(M)}_{2, (j)} + u^{(M)}_{12,(ij)},
\label{eqn:llm}
\end{eqnarray}
where superscripts $I, M$ indicate an independent and a Markov process
respectively. The maximum likelihood estimates (MLE) of the parameters
($u^{(W)}_{b,(c)}$) are available in closed form in literature (Chapters 2.2.3 and
3.1.2 in~\cite{07bfh3a}) and we reproduce them below:
\begin{equation}
\widehat{x_{ij}^{I}}  =  \frac{(x_{i+}) (x_{+j})}{x_{++}} 
\mbox{\hspace{2mm} and \hspace{2mm}}
\widehat{x_{ij}^{M}}  =  x_{ij},
\label{eqn:llmans}
\end{equation}
where $x_{i+}$ and $x_{+j}$ are the sums of the table entries in row
$i$ and column $j$ respectively. $x_{++}$ is the sum of all entries
(i.e., $K/k - 1$, the number of transitions observed in $\{Z^k_t\}$,
or the total number of data points). We compare the fits of the two
models thus: 
\begin{eqnarray}
  \Delta BIC & =  &  BIC^{(I)} - BIC^{(M)} \nonumber \\
             & =  & -2 \sum_{i=0}^{i=1} \sum_{i=0}^{i=1} x_{ij}
               \log\left( \frac{\widehat{x_{ij}^{(I)}}}{x_{ij}}\right)
               - \log\left(\frac{K}{k} - 1\right). \nonumber \\
 \label{eqn:bic}
\end{eqnarray}
Above, we have substituted $\widehat{x_{ij}^{(M)}} = x_{ij}$ and the
fact that the log-linear model for a Markov process has one more
parameter than the independent sampler model.  Large BIC values
indicate a bad fit. A negative $\Delta BIC$ indicates that an
independent model fits better than a Markov model.

This test is applied as follows. We construct a thinned binary
time-series $\{Z^k_t\}$ for $k = |E|$ for each of the edges. The
$\Delta BIC$ is computed and edges with negative $\Delta BIC$ are
deemed to have become independent after $k$ steps of the Markov
chain. The test is repeated with a higher $k$ till all the edges are
deemed to have become independent. A C++ implementation of this test,
using a binary time-series as input, is available
at~\cite{papersoftware}.

However, a difficulty appears at this point. We see that the accuracy
of the log-linear models in Eq.~\ref{eqn:llm} and~\ref{eqn:llmans}
depend on the length of the chain being tested i.e., if the chain is
too short, the ratio of likelihoods and the BIC will be wrongly
calculated. Thus, before we populate the contingency table, we ensure
that the edge time-series is long enough. 

Consider an edge time-series obtained from the $K$-step MC. One can
empirically compute the edge-mean $\overline{Z_t}$ i.e., the
probability of existence of the edge. An estimate of the same quantity can be
obtained from a $k-$thinned version of the chain i.e.,
$\overline{Z^k_t}$. If the thinned time-series are sufficiently long,
then the estimates of the edge mean, $\overline{Z^k_t}$, computed for
different values of $k$, approximately describe a normal distribution,
with mean $q$ (the edge mean) and a variance $\sigma^2$. 

We wish to ensure that a $k$-thinned MC is sufficiently long, i.e.,
the empirical edge-mean $\overline{Z^k_t}$ can be estimated to lie
inside $q \pm r$ with confidence $s$ i.e.$ P( q - r \le
\overline{Z^k_t} < q + r ) = s$, or
\begin{equation}
 \left( \frac{r}{\Phi^{-1}\{ 0.5(1+s) \}} \right)^2 = \sigma^2 
\label{eqn:sigma}
\end{equation}
where $\Phi$ is the cumulative distribution function for a standard
normal distribution and $r$ is a tolerance on the empirical edge-mean
calculated using a $k$-thinned MC.

Consider that thinning the $K$-step time-series by a factor $k^\prime$
renders it a first-order MC (see Appendix~\ref{ap:som} on how this may
be done). The contingency table entries (Eq.~\ref{eqn:llmans}) then
indicate the number (or proportions) of 0-1 and 1-0 transitions of the
2-state Markov chain model of the time-series. It is then trivial to
compute the entries $\alpha$ and $\beta$ of the Markov transition
matrix. If $n^\prime$ is the length of the $k^\prime$-thinned time-series, then
the variance of the $\overline{Z^k_t}$, $\sigma^2$ can be written in
terms of $\alpha, \beta$,
\[ \sigma^2 = \frac{\alpha \beta (2 - \alpha - \beta)}{n^\prime (\alpha +
  \beta)^3} \]
and so, using Eq.~\ref{eqn:sigma},
\begin{equation}
  n^\prime = \frac{\frac{\alpha \beta (2 - \alpha - \beta)}{(\alpha + \beta)^3}}
           {\left( \frac{r}{\Phi^{-1}\{ 0.5(1+s) \}} \right)^2},
\end{equation}
i.e., the $k^\prime$-thinned time-series must be at least $n^\prime$ in
length to provide an estimate $\overline{Z^k_t}$ of $q$ (computed from
the $K$-step time-series) within tolerance $r$ with the required
confidence $s$. In this paper, we use $r = 0.01$ and $s = 0.95$.

We now apply this test to a Markov chain on graphs initiated with
those mentioned in Table~\ref{tab:graphs}. In Fig.~\ref{fig:ind} we
see the thinning factors required to render the $k$-thinned
time-series $\{Z^k_t\}$ resemble independent draws from a
distribution. In Fig.~\ref{fig:ind} (left) we see that thinning by
$4|E|$ is sufficient to render $\{Z^k_t\}$ independent, for all three
graphs tested, when the DD is preserved by the MC. In
Fig.~\ref{fig:ind} (right) we see results for the JDD-preserving
counterpart. It is clear that higher thinning factors are required for
independence, but a thinning factor of $10|E|$ seems to be
sufficient. Thus empirically, the factor-of-two difference in $N$, as
predicted by the models in Eq.~\ref{eqn:N}, and as seen empirically in
Fig.~\ref{fig:dd} and Fig.~\ref{fig:jdd}, are corroborated by a
separate test that does not assume DD or JDD preservation.
\begin{figure*}[t]
  \centerline{
    \includegraphics[width=0.5\textwidth, trim=1cm 6cm 2.cm 6cm,clip=true] {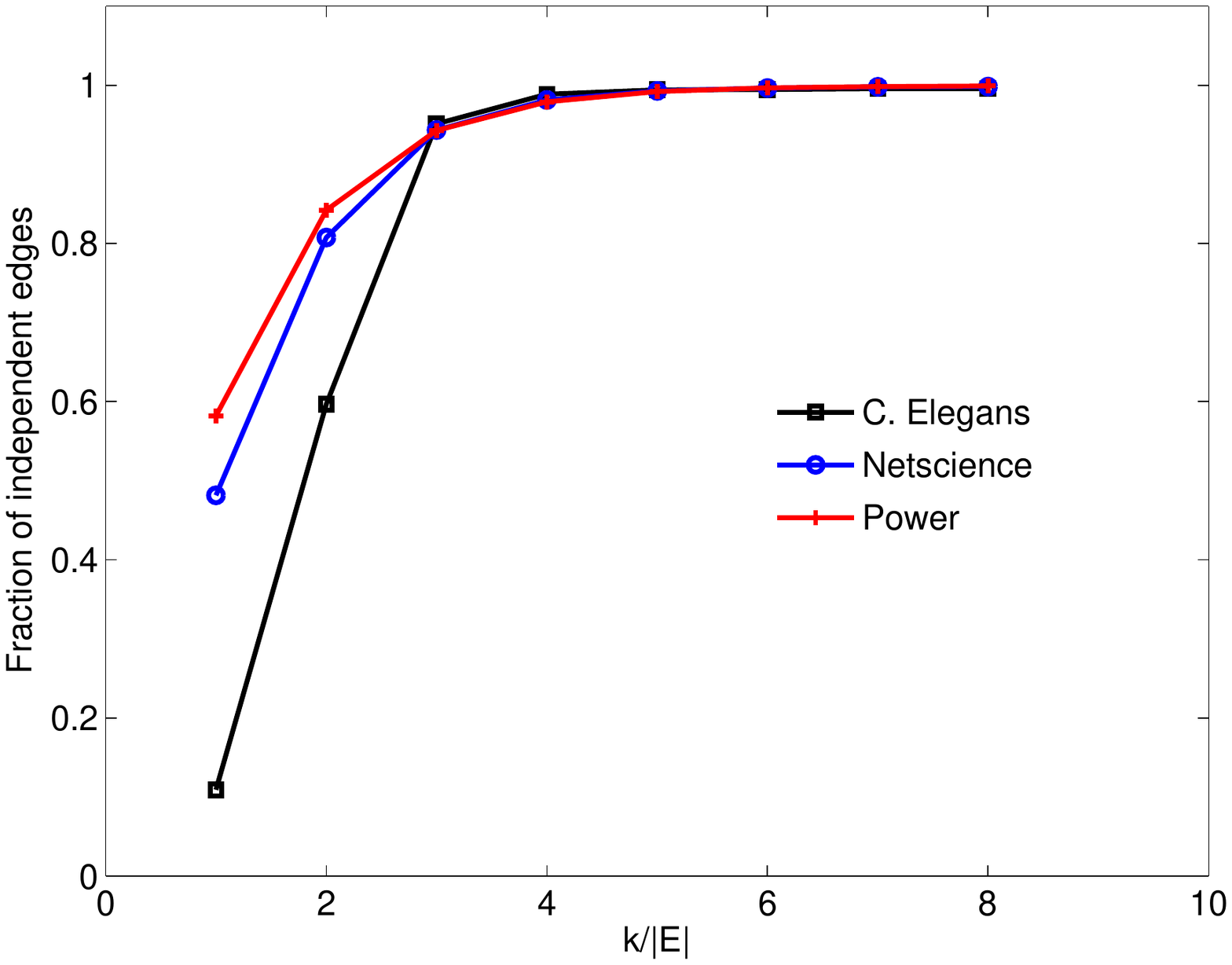}
    \includegraphics[width=0.5\textwidth, trim=1cm 6cm 2.cm 6cm,clip=true] {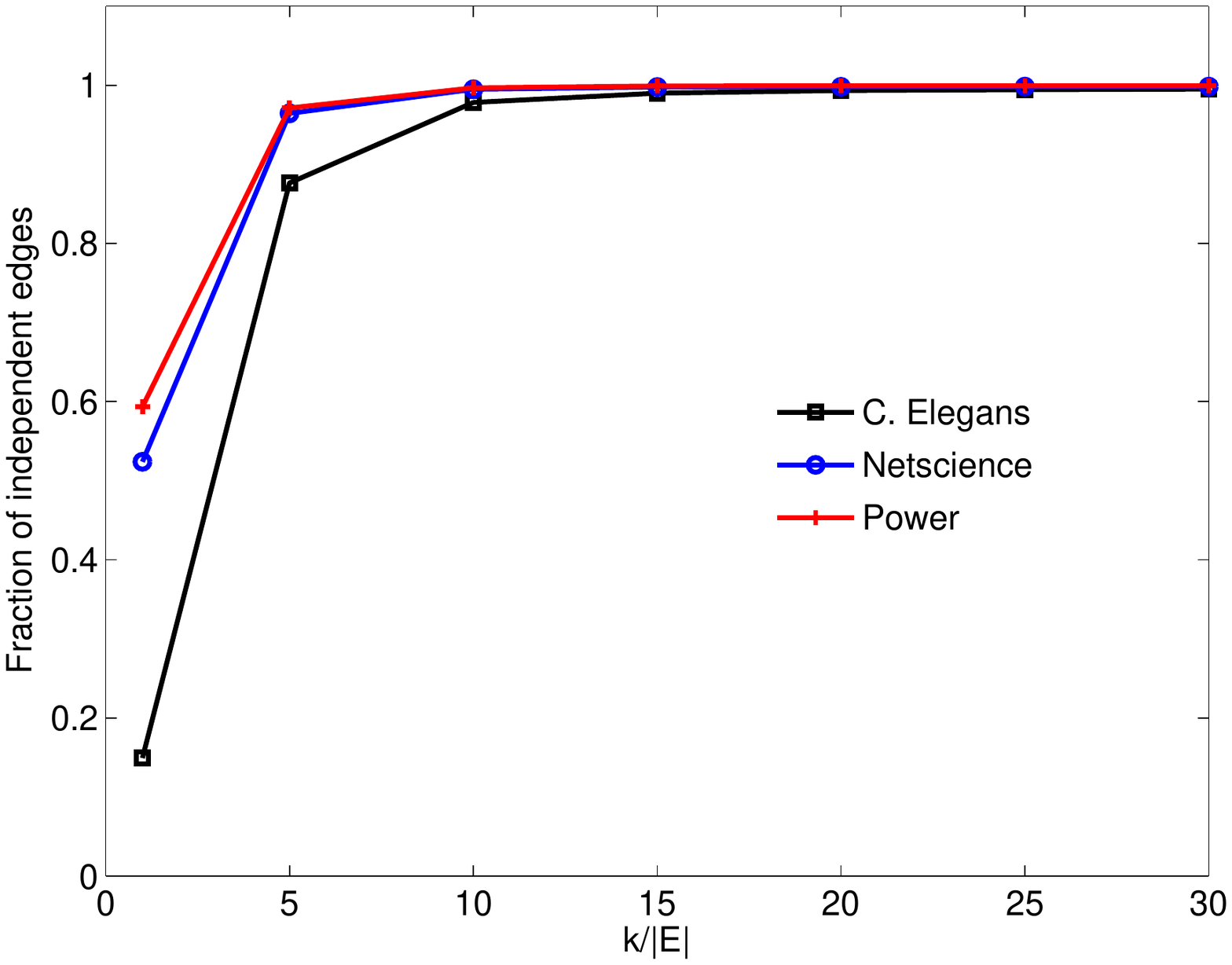}}
  \caption{Fraction of edges testing independent, for ``C. Elegans'',
  ``Netscience'' and ``Power'' for various values of the thinning
  factor $k$. Left: Results for the case when DD is preserved. We see
  that with a thinning factor of $k = 4|E|$ all the edges have
  converged to their stationary distribution. Right: Results for the
  case when JDD was preserved. We see that $k = 10|E|$ ensures that at
  least 95\% of the edges become independent.}
  \label{fig:ind}
\end{figure*}

Convergence of distributions of global clustering coefficients,
diameter and maximum eigenvalues do not automatically indicate that
the time-series of \emph{all} the edges in a graph resemble
independent draws. This is clearly seen in Fig.~\ref{fig:jdd} for the
graph ``Power'' (bottom row). In Fig.~\ref{fig:jdd}, $N \sim 10|E|$
($\epsilon = 4.5\times10^{-5}$) was sufficient to construct
distributions that do not change appreciably for greater $N$. However,
from Fig.~\ref{fig:ind}, we see that about 2\% of the edges are still
not independent, but presumably close to being so, since about $15|E|$
steps of the MC result in all edges becoming independent. 

%

In order to check this issue, we performed a test with the
soc-Epinions1 graph. We generate an ensemble of 1000 graphs using $N =
30|E|$ and the method in Sec.~\ref{sec:JDD} i.e., where the JDD of the
initial graph was preserved as samples were generated. We also run a
long $K$-step MC ($K = 21.6\times10^6|E|$) and compute the thinning
factor $k$ required to render each edge's time-series similar to
independent draws. Since the graph has 405,740 edges, the independence
test was performed for only 10\% of the edges, chosen randomly from
the initial graph. In Fig.~\ref{fig:soc_ind} we plot a histogram of
$k/|E|$. We clearly see that while by $k = 30|E|$ about 90\% of the
edges have become independent, there are quite a few ``pathological''
edges that are very far from becoming independent (one such edge
becomes independent at $k = 720|E|$). 
In Fig.~\ref{fig:soc} we plot the
distributions for the diameter and the maximum eigenvalue for $N =
30|E|, 150|E|, 270|E|$ and $390|E|$. We see that the two
distributions, while not identical, do not change much. Thus
distributions of the chosen metrics are robust (insensitive) to an
ensemble which might contain graphs with a few ``hard-to-converge''
edges. Thus, if one is interested in generating proxy graphs where
anonymity is critical (e.g. proxies of email traffic in an
organization), an exhaustive, edge-by-edge checking of independence,
using the method described above, may be necessary. On the other hand,
when one only desires a collection of roughly independent graphs with
a prescribed DD or JDD, the far simpler approach of Sec.~\ref{sec:DD}
or~\ref{sec:JDD} may suffice.

\begin{figure}
  \includegraphics[width=0.5\textwidth, trim= 1cm 6cm 1cm 6cm, clip=true]{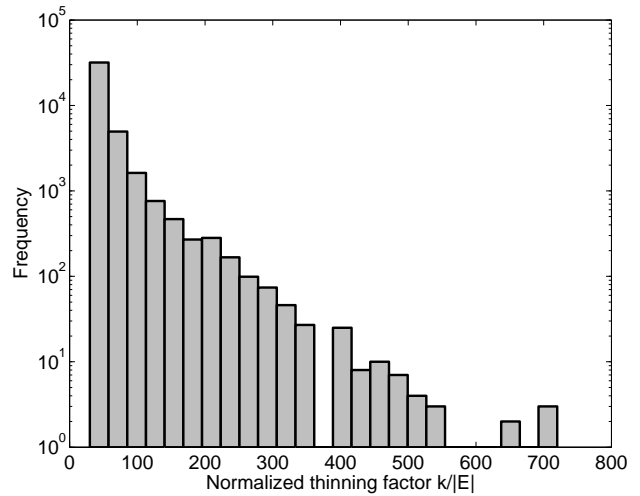}
  \caption{A histogram of $k/|E|$, where $k$ is the number of steps
    the MC has to take to render the time-series of a particular edge
    resemble independent draws, calculated for the soc-Epinions1
    graph. We see that while $30|E|$ steps render about 90\% of the
    edges independent, there are a few edges which are still very far
    from becoming independent.These results are for a case where the
    JDD was preserved.}
\label{fig:soc_ind}
\end{figure}

\begin{figure*}
  \centerline{
    \includegraphics[width=0.5\textwidth, trim= 1cm 6cm 1cm 6cm, clip=true]{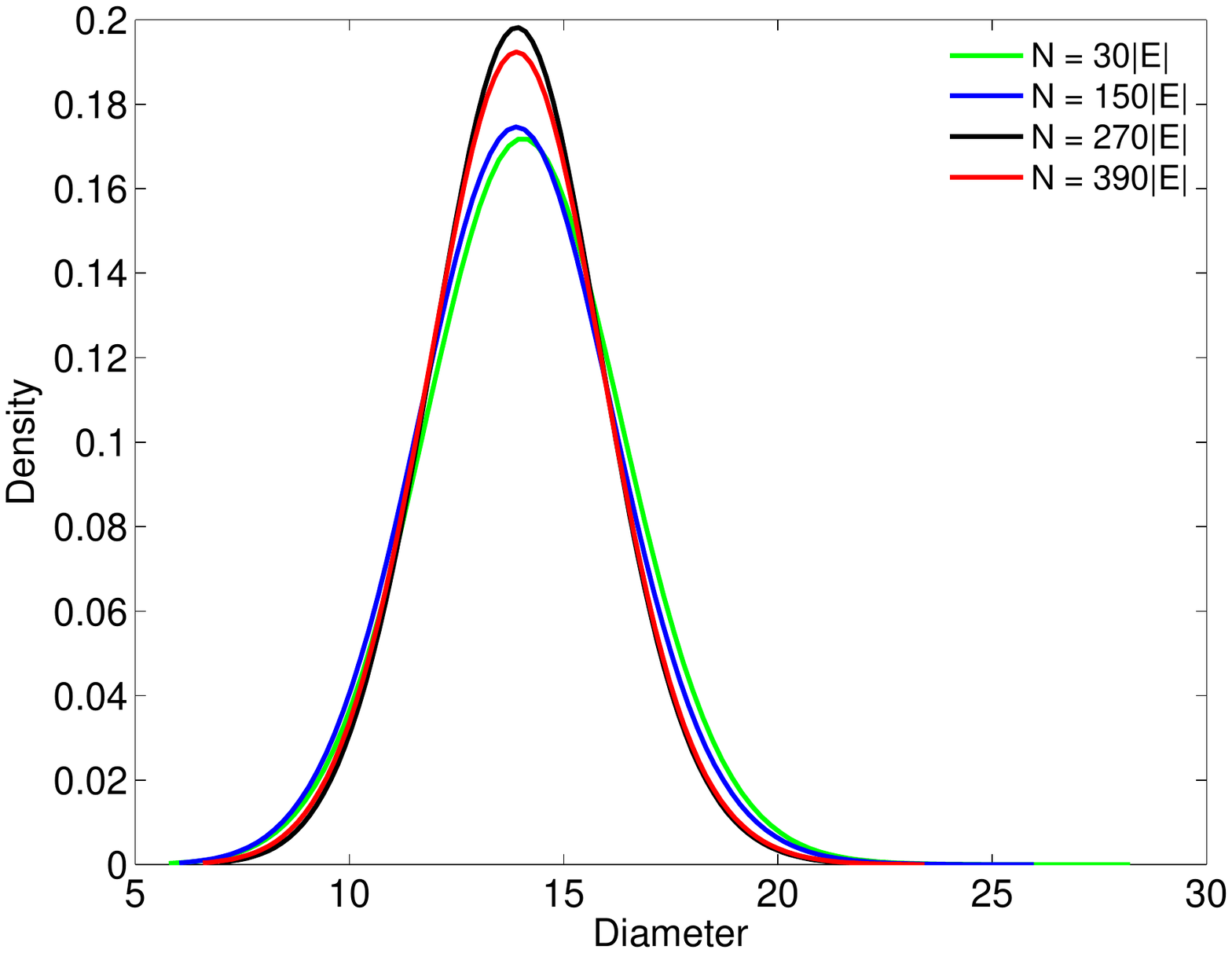}
    \includegraphics[width=0.5\textwidth, trim= 1cm 6cm 1cm 6cm, clip=true]{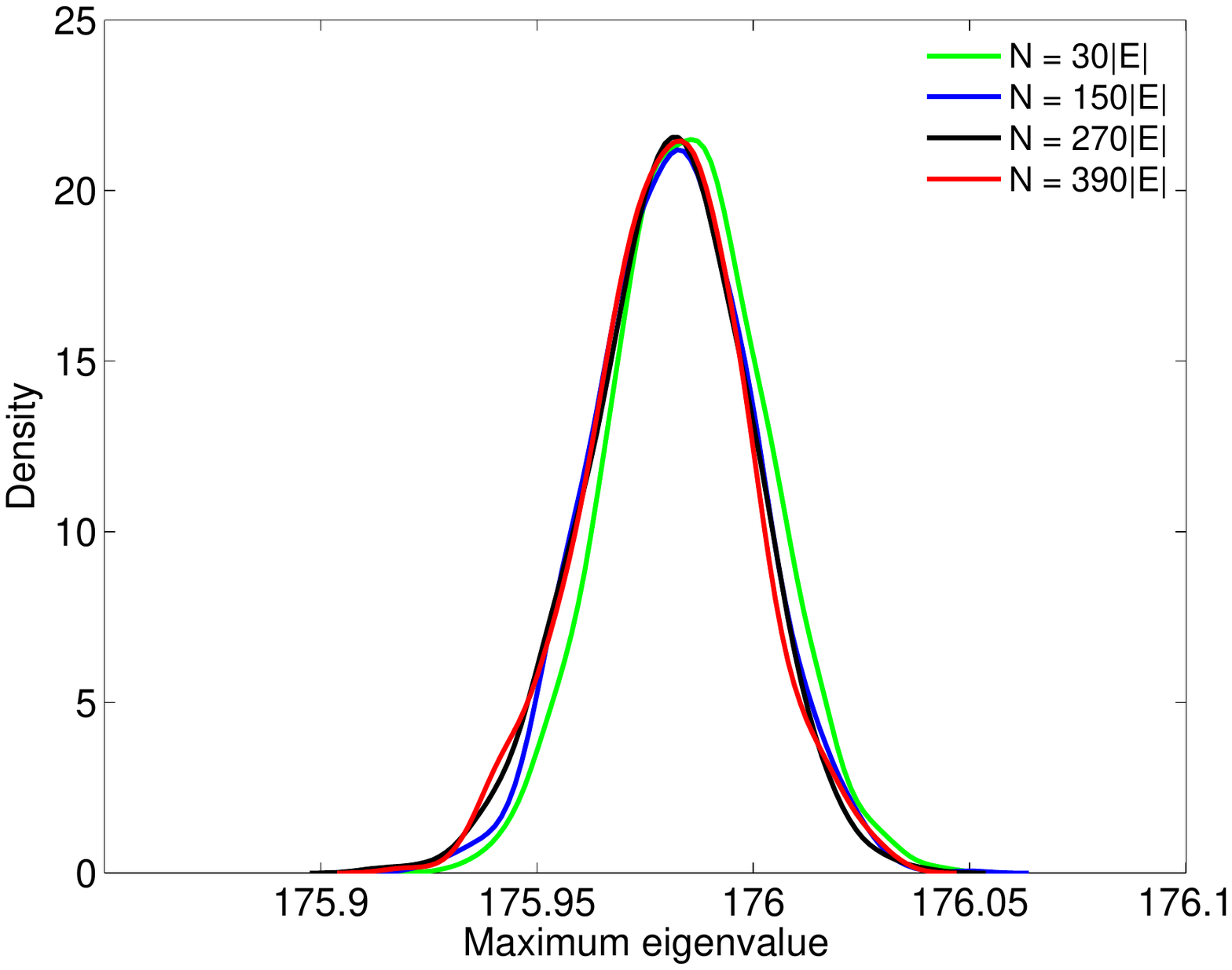}}
  \caption{Plots of the distribution of graph diameter (left) and
  maximum eigenvalue of the graph's Laplacian (right) for $N = 30|E|,
  150|E|, 270|E|$ and $390|E|$, for the soc-Epinions1 graph (green,
  blue, black and red lines respectively). We see that distributions
  are very similar, even though, from Fig.~\ref{fig:soc_ind}, it is
  evident that a few edges are far from being independent. Thus these
  two graphical metrics are not very sensitive graphs that are almost,
  but not quite, independent. These results are for a case where the
  JDD was preserved.}
  \label{fig:soc}
\end{figure*}
\subsection{The Gelman-Rubin test for convergence of Markov chains}

Finally we address the question whether the MC generating the ensemble
of graphs is sampling from a particular mode of the distribution of
graphs (in case the distribution is multimodal). This can, in
principle, occur since the MC is always started from the same starting
graph. Convergence to the stationary distribution (and not just to one
mode) can be tested by starting the MC from an overdispersed set of
points and checking whether the same distributions are realized,
irrespective of the starting point. In our case, we will choose three
points and perform this test only for the JDD-preserving case.

In order to generate a set of overdispersed starting graphs, we run a
Markov chain for $N = 10,000|E|$ steps, after initializing it with a
real network. This serves as the second over-dispersed starting point
for the MC and the initial condition for the third point (realized
after another 10,000$|E|$ steps). Three concurrent MC are initialized
with these graphs, and we calculate the Gelman-Rubin (G-R)
diagnostic~\cite{92gr1a} using the binary edge time-series. Values of
the diagnostic between 1 and 1.1 indicate that the states of the
concurrent Markov chain are not dependent on the starting location. We
performed this test for all 4 graphs; the corresponding G-R
diagnostics are tabulated in Table~\ref{tab:graphs}.

\section{Conclusions}
\label{sec:concl}

We have developed a method for generating independent realizations of
graphs which share a particular graphical property. In this paper, we
have demonstrated the method in the case where (1) the degree
distribution was preserved over all realizations and (2) where the
joint degree distribution was held constant. The graphs are generated
using a MC approach, which drives a graph ``rewiring'' scheme. The
rewiring scheme is responsible for preserving the shared graphical
properties when it generates a new graphical realization, from an old
one. Our method involves running the Markov chain for $N$ steps before
extracting a graph; the chain is run repeatedly to generate an
ensemble of realizations, given an initial graph. We have developed
models and closed form expressions for $N$ that allows a coupled
2-state Markov chain \emph{of an edge} to converge to its stationary
distribution. This is a necessary, but not sufficient, condition for
the Markov chain \emph{on the space of graphs} to ``forget'' the
starting graph and generate an independent realization. We find that a
Markov chain requires $5|E| - 30|E|$ steps to mix fully, i.e., to
provide converged distributions of graphical properties like the global
clustering coefficient, graph diameter and maximum eigenvalues of the
graph Laplacian. We find that the Markov chains that preserve the
degree distribution are easier to mix than those where the joint
degree distribution is held constant.

We verified our model (for $N$) by performing a test for the
independence of the Markov chain described by any edge between two
labeled vertices in the graph, as the Markov chain on the space of
graphs executes its random walk. This test is not dependent on any
heuristic or graphical properties. The test uses the time-series of
occurrence/non-occurrence of an edge between labeled vertices, thins
them by a factor of $N$, and fits it with a first-order Markov and an
independent sampling model. The goodness of fit of the two models,
determined by their individual BICs, is used to select between
them. The Markov chain is considered mixed for any $N$ for which the
independence sampling model is selected over the Markov model. We find
that even for $N$ that provide converged distributions of the
graphical properties above, a small fraction of the edges may still be
significantly correlated (as determined by the independence test). The
robustness of the graphical properties indicate that if fully
independent graphical realizations are desired, the laborious process
of testing each edge may be unavoidable.

Finally, we repeated our tests by starting parallel Markov chains from
widely dispersed starting points and computing the Gelman-Rubin
diagnostic for them. We find that the parallel chains converge to the
same stationary distribution. This check was performed to ensure that
our distributions of graphical properties was not being driven by the
starting point of the Markov chain.

Apart from the theoretical results summarized above, we also provide
an open-source C++ implementation of the non-parametric test, along
with sample problems to illustrate its use~\cite{papersoftware}.

\section*{Acknowledgments}
This work was funded by the Applied Mathematics program of the United States Department of Energy, Office
of Science and by an Early Career Award from the Laboratory Directed
Research \& Development (LDRD) program at Sandia National
Laboratories. Sandia National Laboratories is a multiprogram
laboratory managed and operated by Sandia Corporation, a wholly owned
subsidiary of Lockheed Martin Corporation, for the United States
Department of Energy's National Nuclear Security Administration under
contract DE-AC04-94AL85000.

\appendix

\section{}
\label{ap:som}

Here, we address how one may thin a long, $K$-step time-series so that
the thinned chain resembles a first-order Markov process. It follows
the arguments in Sec.~\ref{sec:ebe}. We consider a binary time-series
and populate a $2 \times 2 \times 2$ contingency table of transitions
that one might observe in a sequence of length three. The contingency
table is fitted with log-linear models for second- and first-order
Markov models (see Sec. 7.3.1 in~\cite{07bfh3a} for second-order
Markov models), and the $G^2$ statistic calculated along with the
difference in the BICs arising from second-order and first-order
Markov model fits.

Initially, the high-order Markov model fits to data better than the
first-order model. However, as the $K$-step time-series is thinned,
and the correlation decays, the fit of first-order model,
vis-\'{a}-vis, the second-order model improves. At a particular
thinning factor $k^\prime$ the first-order model fits better than the
second order one (the difference in BICs favors the first-order model)
and we obtain a first-order Markov chain.

In the software that we have released with this
paper~\cite{papersoftware}, the C function \texttt{mctest()} checks
the fit of a time-series to first- and second-order Markov chains,
returning the ratio of likelihoods and the difference in the BICs of
the two models. The C function \texttt{indtest()} does the same, but
between  a first-order Markov model and independent draws from a
binary distribution. Both the functions construct the contigency table
from the binary time-series supplied via the function arguments, and
compute the expected values of the table entries (using log-linear
models calibrated to table data). In contrast, the function
\texttt{mcest()} estimates the transition probabilities $\alpha$ and
$\beta$ from a binary time-series, being modeled as a first-order
Markov process. Examples of the use of the tests (\texttt{mctest} and
\texttt{indtest}) are illustrated in the software package associated
with this paper (in \texttt{ex01/}). The use of all the functions, in
a MC on graphs, is in \texttt{ex02/}.



\begin{thebibliography}{10}

\bibitem{Aiello00}
{\sc W.~Aiello, F.~Chung, and L.~Lu}, {\em A random graph model for power law
  graphs}, Experimental Math, 10 (2000), pp.~53--66.

\bibitem{BaWe00}
{\sc A.~Barrat and M.~Weigt}, {\em On the properties of small-world network
  models}, The European Physical Journal B - Condensed Matter and Complex
  Systems, 13 (2000), pp.~547--560.

\bibitem{BaKiSa10}
{\sc M.~Bayati, J.~H. Kim, and A.~Saberi}, {\em A sequential algorithm for
  generating random graphs}, Algorithmica, 58 (2010), pp.~860--910.

\bibitem{07bfh3a}
{\sc Y.~M. Bishop, S.~E. Fienberg, and P.~W. Holland}, {\em Discrete
  multivariate analysis: {T}heory and practice}, Springer-Verlag, New York, NY,
  2007.

\bibitem{BlDi11}
{\sc J.~Blitzstein and P.~Diaconis}, {\em A sequential importance sampling
  algorithm for generating random graphs with prescribed degrees}, Internet
  Mathematics, 6 (2011), pp.~489--522.

\bibitem{Br86}
{\sc A.~Z. Broder}, {\em How hard is it to marry at random? (on the
  approximation of the permanent)}, in Proceedings of the eighteenth annual ACM
  symposium on Theory of computing, STOC '86, New York, NY, USA, 1986, ACM,
  pp.~50--58.

\bibitem{Chung2}
{\sc F.~R.~K. Chung and L.~Lu}, {\em The average distance in a random graph
  with given expected degrees}, Internet Mathematics, 1 (2003).

\bibitem{92gr1a}
{\sc A.~Gelman and D.~B. Rubin}, {\em Inference from iterative simulation using
  multiple sequences}, Statistical Science, 7 (1992), pp.~457--472.

\bibitem{12gk3a}
{\sc M.~{Gjoka}, M.~{Kurant}, and A.~{Markopoulou}}, {\em {2.5K-Graphs: from
  Sampling to Generation}}, ArXiv e-prints,  (2012).

\bibitem{GkantsidisMMZ03}
{\sc C.~Gkantsidis, M.~Mihail, and E.~W. Zegura}, {\em The {M}arkov chain
  simulation method for generating connected power law random graphs}, ALENEX,
  (2003), pp.~16--25.

\bibitem{62hs1a}
{\sc S.~L. Hakimi}, {\em On the realizability of a set of integers as degrees
  of the vertices of a simple graph}, J. SIAM Appl. Math., 10 (1962),
  pp.~496--506.

\bibitem{55hv1a}
{\sc V.~Havel}, {\em A remark on the existence of finite graphs}, \u{C}asopis.
  P\u{e}st. Mat., 80 (1995), pp.~477--480.

\bibitem{JeSiVi04}
{\sc M.~Jerrum, A.~Sinclair, and E.~Vigoda}, {\em A polynomial-time
  approximation algorithm for the permanent of a matrix with nonnegative
  entries}, J. ACM, 51 (2004), pp.~671--697.

\bibitem{ktv}
{\sc R.~Kannan, P.~Tetali, and S.~Vempala}, {\em Simple {M}arkov-chain
  algorithms for generating bipartite graphs and tournaments}, Random Struct.
  Algorithms, 14 (1999), pp.~293--308.

\bibitem{kronecker}
{\sc J.~Leskovec, D.~Chakrabarti, J.~Kleinberg, C.~Faloutsos, and
  Z.~Ghahramani}, {\em Kronecker graphs: An approach to modeling networks},
  Journal of Machine Learning Research (JMLR), 11 (2010), pp.~985--1042.

\bibitem{Me94}
{\sc R.~Merris}, {\em Laplacian matrices of graphs: a survey}, Linear Algebra
  and its Applications, 197-198 (1994), pp.~143 -- 176.

\bibitem{newman_graphs}
{\sc M.~E.~J. Newman}, {\em Prof. {M}.~{E}.~{J}. {N}ewman's collection of
  graphs at {U}niversity of {M}ichigan}.
\newblock http://www-personal.umich.edu/$\sim$mejn/netdata/.

\bibitem{adjnoun}
{\sc M.~E.~J. Newman}, {\em Finding community structure in networks using the
  eigenvectors of matrices, 036104}, Phys. Rev. E, 74 (2006).

\bibitem{12pf4a}
{\sc J.~J. {Pfeiffer}, III, T.~{La Fond}, S.~{Moreno}, and J.~{Neville}}, {\em
  {Fast Generation of Large Scale Social Networks with Clustering}}, ArXiv
  e-prints,  (2012).

\bibitem{96ra1a}
{\sc A.~Raftery and S.~M. Lewis}, {\em Implementing {MCMC}}, in {M}arkov Chain
  {M}onte {C}arlo in Practice, W.~R. Gilks, S.~Richardson, and D.~J.
  Spiegelhalter, eds., Chapman and Hall, 1996, pp.~115--130.

\bibitem{92rl2a}
{\sc A.~E. Raftery and S.~M. Lewis}, {\em How many iterations in the {G}ibbs
  sampler?}, in {B}ayesian {S}tatistics, J.~M. Bernardo, J.~O. Berger, A.~P.
  Dawid, and A.~F.~M. Smith, eds., vol.~4, Oxford University Press, 1992,
  pp.~765--766.

\bibitem{papersoftware}
{\sc J.~Ray, A.~Pinar, and C.~Safta}, {\em graph{MC}: A package for testing the
  independence of graphs}.
\newblock http://www.sandia.gov/$\sim$apinar/graph{MC}/graph{MC}.html.

\bibitem{03ra3a}
{\sc M.~Richardson, R.~Agrawal, and P.~Domingos}, {\em Trust management for the
  semantic {W}eb}, in The Semantic Web - ISWC 2003, D.~Fensel, K.~Sycara, and
  J.~Mylopoulos, eds., vol.~2870 of Lecture Notes in Computer Science, Springer
  Berlin / Heidelberg, 2003, pp.~351--368.
\newblock 10.1007/978-3-540-39718-2\_23.

\bibitem{12sk3a}
{\sc C.~Seshadhri, T.~G. Kolda, and A.~Pinar}, {\em Community structure and
  scale-free collections of {E}rd{\"{o}}s-{R}\'enyi graphs}, Phys. Rev. E, 85
  (2012), p.~056109.

\bibitem{Sokal}
{\sc A.~Sokal}, {\em {M}onte {C}arlo methods in statistical mechanics:
  {F}oundations and new algorithms}, 1996.

\bibitem{snap}
{\sc {S}tanford {N}etwork {A}nalysis {P}latform Collection~of Graphs}, {\em The
  {E}pinions social network from the {S}tanford {N}etwork {A}nalysis {P}latform
  collection}.
\newblock http://snap.stanford.edu/data/soc-Epinions1.html.

\bibitem{StPi12}
{\sc I.~Stanton and A.~Pinar}, {\em Constructing and sampling graphs with a
  prescribed joint degree distribution using {M}arkov chains}, ACM Journal of
  Experimental Algorithmics.
\newblock to appear.

\bibitem{81tr1a}
{\sc R.~Taylor}, {\em Constrained switchings in graphs}, in Combinatorial
  Mathematics VIII, K.~McAvaney, ed., vol.~884 of Lecture Notes in Mathematics,
  Springer Berlin / Heidelberg, 1981, pp.~314--336.
\newblock 10.1007/BFb0091828.

\bibitem{96wp2a}
{\sc S.~Wasserman and P.~E. Pattison}, {\em Logit models and logistic
  regression for social networks: I. {A}n introduction to {M}arkov graphs and
  $p^*$}, Psychometrika, 61 (1996), pp.~401--425.

\bibitem{celegans1}
{\sc D.~J. Watts and S.~H. Strogatz}, {\em Collective dynamics of 'small-world'
  networks}, Nature, 393 (1998), pp.~440--442.

\end{thebibliography}
\end{document}